\documentclass[conference]{IEEEtran}
\ifCLASSINFOpdf
\else
\fi
\usepackage{graphicx}
\usepackage{array}
\usepackage{amsmath}
\usepackage{float}
\usepackage{float}
\usepackage{multirow}
\usepackage{array}
\usepackage{cancel}
\usepackage[normalem]{ulem}
\usepackage{tabu}
\usepackage{amsthm}
\usepackage{epsfig}
\usepackage{epstopdf}
\usepackage{epsf}
\usepackage{color}
\usepackage[english]{babel}

\newtheorem{proposition}{Proposition}

\begin{document}
\medskip
\title{Inverse of a Special Matrix and Application}

\author{\IEEEauthorblockN{Thuan Nguyen}
\IEEEauthorblockA{School of Electrical and Computer Engineering, Oregon State University, Corvallis, OR, 97331\\
Email: nguyeth9@oregonstate.edu}
%\and
%\IEEEauthorblockN{Duong Nguyen-Huu}
%\IEEEauthorblockA{School of Electrical and\\Computer Engineering\\
%Oregon State University\\
%Corvallis, OR, 97331\\
%Email: nguyendu@eecs.oregonstate.edu}
%\and
%\IEEEauthorblockN{Thinh Nguyen}
%\IEEEauthorblockA{School of Electrical and\\Computer Engineering\\
%Oregon State University\\
%Corvallis, 97331 \\
%Email: thinhq@eecs.oregonstate.edu}
}


%Telephone: (800) 555--1212, Fax: (888) 555--1212}


\maketitle

\begin{abstract}
The matrix inversion is an interesting topic in algebra mathematics. However, to determine an inverse matrix from a given matrix is required many computation tools and time resource  if the size of matrix is huge. In this paper, we have shown an inverse closed form for an interesting matrix which has much applications in communication system. Base on this inverse closed form, the channel capacity closed form of a communication system can be determined via the error rate parameter $\alpha$. 
\end{abstract}


% \fi
\IEEEpeerreviewmaketitle

\section{Matrix Construction}

In Wireless communication system or Free Space Optical communication system, due to the shadow effect or the turbulent of environment, the channel conditions can be flipped from ``good" to ``bad" or ``bad" to ``good" state such as Markov model after the transmission time $\sigma$ \cite{mcdougall2003sensitivity} \cite{wang1995finite}. For simple intuition, in ``bad" channel, a signal will be transmitted incorrectly and in ``good" channel, the signal is received perfectly. Suppose a system has total $n$ channels, the ``good" channel is noted as ``1" and ``bad" channel is ``0", respectively, the transmission time between transmitter and receiver is $\sigma$, the probability the channel is flipped after the transmission time $\sigma$ is $\alpha$. We note that if the system using a binary code such as On-Off Keying in Free Space Optical communication, then the flipped probability $\alpha$ is equivalent to the error rate.

Consider a simple case for $n=2$, suppose that at the beginning, both channel is ``good" channel, the probability of system has both of channels are ``good" after transmission time $\sigma$, for example, is $(1-\alpha)^2$. Let call $A_{ij}$ is the probability of system from the state has $i-1$ ``good" channels and $n-i+1$ ``bad" channels transfers to state has $j-1$ ``good" and $n-j+1$ ``bad" channels. Obviously that $1 \leq i \leq n+1$ and $1 \leq j \leq n+1$. For example, the  transition matrix $A_2$ and $A_3$ for $n=2$ and $n=3$ are constructed respectively as follows:\\

\begin{small}
$\begin{array}{cc}
A_2 = \begin{bmatrix}
(1-\alpha)^2 &2\alpha(1-\alpha) & \alpha^2\\
\alpha(1-\alpha) &(1-\alpha)^2+\alpha^2 & \alpha(1-\alpha)\\
\alpha^2  & 2\alpha(1-\alpha) & (1-\alpha)^2 \\
\end{bmatrix}.
\end{array} $
\end{small}\\

\begin{tiny}
$\begin{array}{cc}
A_3 = \begin{bmatrix}
(1-\alpha)^3 &3(1-\alpha)^2\alpha & 3\alpha^2(1-\alpha) & \alpha^3\\
(1-\alpha)^2\alpha &   2(1-\alpha)\alpha^2+(1-\alpha)^3 & 2(1-\alpha)^2\alpha+\alpha^3 & (1-\alpha)\alpha^2\\
(1-\alpha)\alpha^2 &  2(1-\alpha)^2\alpha+\alpha^3 & 2(1-\alpha)\alpha^2+(1-\alpha)^3 & (1-\alpha)^2\alpha\\
\alpha^3 & 3\alpha^2(1-\alpha)& 3(1-\alpha)^2\alpha & (1-\alpha)^3\\
\end{bmatrix}.
\end{array}$
\end{tiny}

These transition matrices are obviously size $(n+1)\times(n+1)$ since the number of ``good" channels can achieve $n+1$ discrete values from $0,1,\dots,n$. Moreover, these class matrices have several interesting properties: (1) all entries in matrix $A_n$ can be determined by Proposition \ref{prop:SRC}; (2) the inverse of matrix $A_n$ is given by Proposition \ref{inverse}. Moreover, this matrices are obviously central symmetric matrix.

\begin{proposition}
\label{prop:SRC}
For $n$ channels system, the transition matrix $A_n$ has size $(n+1) \times (n+1)$ and all entries ${A_n}_{ij}$ in row $i$ column $j$ will be established by
{\tiny $${A_n}_{ij}=\sum_{s=\max(i-j,0)}^{s=\min(n+1-j,i-1)}{j-i+s\choose n+1-i} {s\choose i-1} \alpha^{j-i+2s} (1-\alpha)^{n-(j-i+2s)}$$  }
\end{proposition}

\begin{proof}
From the definition, ${A_n}_{ij}$ is the probability from state has $i-1$ ``good" channels or $i-1$ bit ``1"  transfer to state  has $j-1$ ``good" channels or $j-1$ bit ``1". Therefore, suppose $s$ is the number channels in $i-1$ ``good" channels that is flipped to ``bad" channels  after the transmission time $\sigma$ and $0 \leq s \leq i-1$. Thus, to maintain $j-1$ ``good" channels after the time $\sigma$, the number of ``bad" channels in $n+1-i$ ``bad" channels must be flipped to ``good" channels is:
$$(j-1)-((i-1)-s)=j-i+s$$

Therefore, the total number of channels are flipped their state after transmission time $\sigma$ is:
$$s+(j-i+s)=j-i+2s$$ 
 and the total number of channels that preserves their state after transmission time $\sigma$ is  $n-(j-i+2s)$. However, $0 \leq s \leq i-1$. Similarly, the number of ``bad" channels in $n+1-i$ ``bad" channels must be flipped to ``good" channels should be in $0 \leq j-i+s \leq n+1-i$. Hence:
 
$$
\begin{cases}
\max {s} = \min (n+1-j;i-1 ) \\
\min {s} = \max (0;i-j )
\end{cases}
$$

Therefore, ${A_n}_{ij}$ can be determined by below form:

{\tiny $${A_n}_{ij}=\sum_{s=\max(i-j,0)}^{s=\min(n+1-j,i-1)}{j-i+s\choose n+1-i} {s\choose i-1} \alpha^{j-i+2s} (1-\alpha)^{n-(j-i+2s)}$$  }

\end{proof}

\begin{proposition}
\label{inverse}
All the entries of inverse matrix $A_n^{-1}$ given in Proposition \ref{prop:SRC} can be determined via original transition matrix ${A_n}$ for $\forall$ $\alpha \neq 1/2$.

 $${A_n}_{ij}^{-1}=  \dfrac{(-1)^{i+j}}{(1-2\alpha)^{n}} {A_n}_{ij}$$

\end{proposition}

Due to the pages limitation, we will show the detailed proof at the end of this paper.  To illustrate our result, an example of the inverse matrix $A_2$ are shown as follows:\\

\begin{small}
$\begin{array}{cc}
A_2^{-1} =\dfrac{1}{(1-2\alpha)^2} \begin{bmatrix}
(1-\alpha)^2 & -2\alpha(1-\alpha) & \alpha^2\\
-\alpha(1-\alpha) &(1-\alpha)^2+\alpha^2 & -\alpha(1-\alpha)\\
\alpha^2  & -2\alpha(1-\alpha) & (1-\alpha)^2 \\
\end{bmatrix}
\end{array}. $
\end{small}

Next, base on the existence of inverse matrix closed form, we will show that a capacity closed form for a  discrete memory-less channel can be established. We note that in \cite{cover2012elements}, the authors said that haven't closed form for channel capacity problem. However, with our approach, the closed form can be established for a wide range of channel with error rate $\alpha$ is small.

\section{Optimize system capacity} 
A discrete memoryless channel is characterized by a channel matrix $A \in \mathbf{R}^{m \times n}$ with $m$ and $n$ representing the numbers of distinct input (transmitted) symbols $x_i$, $i = 1, 2, \dots, m$,  and output (received) symbols $y_j$, $j = 1, 2, \dots, n$, respectively.  
The matrix entry $A_{ij}$ represents the conditional probability that given a symbol $x_i$ is transmitted, the symbol $x_j$ is received.  Let $p = (p_1, p_2, \dots, p_m)^T$ be the input probability mass vector, where $p_i$ denotes the probability of transmitting symbol $x_i$, then the probability mass vector of output symbols $q = (q_1, q_2, \dots, q_n)^T = A^Tp$, where $q_i$ denotes the probability of receiving symbol $y_i$. For simplicity, we only consider the case $n=m$ such that the number of transmitted input patterns is equal the number of received input patterns. The mutual information between input and output symbolsis:
\begin{equation*}
I(X;Y) = H(Y) - H(Y|X),
\end{equation*}
where
\begin{eqnarray*}
 H(Y) &=&-\sum_{j=1}^{j=n}{q_{j}\log{q_j}} \\
H(Y|X) &=& \sum_{i=1}^{m}\sum_{j=1}^{n} {p_i A_{ij}} \log {A_{ij}}.
\end{eqnarray*}
Thus, the mutual information function can be written as:
\begin{equation*}
I(X;Y) = -\sum_{j=1}^{j=n}{(A^Tp)_j\log{(A^Tp)_j}} +\sum_{i=1}^{m}\sum_{j=1}^{n} {p_iA_{ij}} \log {A_{ij}},
\end{equation*}
where $(A^Tp)_j$ denotes the $j$th component of the vector $q = (A^Tp)$.  The capacity $C$ of a discrete memoryless channel associated with a channel matrix $A$ puts a theoretical maximum rate that information can be transmitted over the channel \cite{cover2012elements}.  It is defined as: 
\begin{equation}
\label{eq: capacity}
C = \max_{p}{I(X;Y)}.
\end{equation}
Therefore, finding the channel capacity is to find an optimal input probability mass vector $p$ such that the mutual information between the input and output symbols is maximized. For a given channel matrix $A$, $I(X;Y)$ is a concave function in $p$ \cite{cover2012elements}.  Therefore, maximizing $I(X;Y)$ is equivalent to minimizing $-I(X;Y)$, and the capacity problem can be cast as the following convex problem:

\indent Minimize: 
 \begin{equation}
	\sum_{j=1}^{n}{(A^Tp)_{j}\log{(A^Tp)_j}} -\sum_{i=1}^{m}\sum_{j=1}^{n} {p_i A_{ij}} \log {A_{ij}} \nonumber \\
\end{equation}
\indent Subject to: 
 $$\begin{cases}
& p_i \succeq \mathbf{0}\\
& \mathbf{1}^Tp = 1 \\
\end{cases}$$

Optimal numerical values of $p^*$ can be found efficiently using various algorithms such as gradient methods \cite{grant2008cvx} \cite{boyd2004convex}.  However, in this paper, we try to figure out the closed form for optimal distribution $p$ via KKT condition. The KKT conditions state that for the following canonical optimization problem:

Problem 
\indent Miminize: $f(x)$ \\
\
\indent Subject to: 

$$g_i(x) \le  0, i = 1, 2, \dots n,$$
$$h_j(x) = 0, j = 1, 2, \dots, m, $$

construct the Lagrangian function:
\begin{equation}
\label{sec:lagrangian}
L(x,\lambda, \nu) = f(x) + \sum_{i=1}^n{\lambda_i g_i(x)} + \sum_{j=1}^m{\nu_j h_j(x)},
\end{equation}

then for $i = 1, 2, \dots, n$, $j = 1, 2, \dots, m$, the optimal point $x^*$ must satisfy:

\begin{equation}
\label{eq:kkt1}
\begin{cases}
%\begin{equation} 
%\label{eq:kkt1}
g_i(x^*) \le 0, \\
%\end{equation}  \\
h_j(x^*) = 0, \\
\frac{d{L(x, \lambda, \nu)}}{dx}|_{x = x^*, \lambda = \lambda^*, \nu = \nu^*} = 0, \\
\lambda_i^*x_i^* = 0, \\
\lambda_i^* \ge 0.
\end{cases}
\end{equation}

Our transition matrix that is already established in previous part can represent as a channel matrix.  In the optical transmission, for example, the transmission bits are denoted by the different levels of energy, for example, in On-Off Keying code bit ``1" and ``0" is represented by high and low power level. This energy is received by a photo diode and converse directly to the voltage for example. However, these photo diode work base on the aggregate property when collecting all the incident energy, that said, if two channels transmit a bit ``1" then the photo diode will receive the same energy ``2" even though this energy comes from a different pair of channels. Therefore, the received signal is completely dependent  to the number of bits ``1" in transmission side. Hence, in receiver side, the photo diode will recognize $n+1$ states $0,1,2,\dots,n$. From this property, the transition matrix $A$ is the previous section is exactly the system channel matrix. The channel capacity of system, therefore, is determined as an optimization problem in (\ref{eq: capacity}). 

Next, we will show that the above optimization problem can be solved efficiently by KKT condition. We note that our method can establish the closed form for general channel matrix and then the results are applied to special matrix $A_n$. First, we try to optimize directly with input distribution $p$, however, the KKT condition for input distribution is too complicated to construct the first derivation. On the other hand, base on the existence of inverse channel matrix, the output variable is more suitable to work with KKT condition since.  Due to $0 \leq q_j \leq 1$, the Lagrange function from (\ref{sec:lagrangian}) with output variable $q$ is:

$$L(q_j,\lambda_j,\nu_j)= I(X,Y) + \sum_{j=1}^{j=n}{{q_j}{\lambda_j}} + \nu (\sum_{j=1}^{j=n}{q_j}-1) $$ 

Using KKT conditions, at optimal point $q_{j}^{*}$, $\lambda_{j}^{*}$, $\nu^{*}$:  

$$ 
\begin{cases}
q_{j}^{*} \geq 0\\
\sum_{j=1}^{j=n}q_{j}^{*} = 1\\
\nu^{*} - \lambda_{j}^{*} - \dfrac{dI(X,Y)}{dq_{j}^{*}}= 0\\
\lambda_{j}^{*} \geq 0\\
\lambda_{j}^{*} q_{j}^{*}=0
\end{cases}
$$

Because $0 \leq p_i \leq 1, i=1,\dots, (n)$ and $\sum_{i=1}^{n}{p_i} = 1$, so always exist $p_i> 0$. From $q_j = \sum_{i=1}^{i=n}p_i{A}_{ij}$ with $\forall$ $A_{ij}>0$, we can see clearly that $q_j > 0$ with $\forall q_j$ or $q_{j}^{*} > 0$ with $\forall q_j^{*}$.

Therefore with fifth condition, $\lambda_{j}^{*} = 0$ with $\forall \lambda_j^{*}$. Then, we have simplified KKT conditions:
$$
\begin{cases}
\sum_{j=1}^{j=n}q_{j}^{*} = 1\\
\nu^{*} - \dfrac{dI(X,Y)}{dq_{j}^{*}}= 0\\
\end{cases}
$$

The derivations are determined by:

$$\dfrac{dI(X,Y)}{dq_{j}} = \sum_{i=1}^{i=n} {A}_{ji}^{-1} \sum_{j=1}^{j=n} {A}_{ij}\log{A}_{ij} - (1+\log{q_j})$$

Let call:  $$\sum_{i=1}^{i=n} {{A}_{ji}^{-1}} \sum_{j=1}^{j=n} {A}_{ij}\log{A}_{ij}={K}_{j}$$

Next, using derivation of I(X,Y) at $q_{j}=q_j^*$ and last $KKT$ condition:
$$\nu^*={K}_j -(1+\log{q_j}^*)$$

Hence:
$$q_j^* = 2^{{K}_j-\nu^* -1}$$

Next, using first $KKT$ simplified condition, we have the sum of all output states is 1. 
$$\sum_{j=1}^{j=n} 2^{{K}_j-\nu^*-1} = 1$$
%$$2^{-\nu^*} \sum_{j=0}^{j=n}{2^{{K_n}_j-1}} =1 $$
$$2^{\nu^*} = \sum_{j=1}^{j=n}{2^{{K}_j-1}}$$

Therefore, $\nu^*$ can be figured out by:
$$\nu^*= \log{\sum_{j=1}^{j=n}{2^{{K}_j-1}}}$$

From the second $KKT$ simplified condition, we can compute $\forall$ $q_j^{*}$:
$$q_j^* = 2^{{K}_j-\nu^* -1}$$

And finally:
$${p^{T}}^{*}={q^{T}}^{*}{A}_{ij}^{-1}$$
 
Due to the channel matrix is a closed form of $\alpha$, the optimal input vector $p$ and output vector $q$ also is a function of $\alpha$. However, we note that since the KKT condition works directly to the output variable $q$, the optimal input $p$ can be invalid $p_i> 1$ or $p_i <0$. In next step, our simulations shown that for  $n \leq 10$ and $\alpha \leq 0.2$, both output and input vector are valid. That said, our approach will be worked with a good system where the error probability $\alpha$ is small. In case of the invalid optimal input vector, the upper bound of channel capacity, of course, will be established.

\section{Conclusion}

 In this paper, our contributions are twofold: (1) establish an inverse closed form for a class of channel matrix based on the error probability $\alpha$; (2) figure out the closed form for channel matrix with small error rate $\alpha$ and determine the upper bound system capacity for a high error rate channel. 
\medskip
\bibliographystyle{unsrt}
\bibliography{sample}

\appendix
\textbf{Proof for Proposition \ref{inverse}.}
\begin{proof}
To simplify our notation, the ``good" and ``bad" channel are represented by bit ``1" and ``0", respectively. Next, we will use the definition to show that: 
$${A_n}  {A_n}^{-1}=I$$ 

If matrix $A_n^*$ is constructed by ${A_n^*}_{ij}=(-1)^{i+j}{A_n}_{ij}$, then we need to show that: 
$${A_n}  {A_n^*} =B = (1-2\alpha)^n I$$

Firstly, we note that the ${A_n}_{ij}$ and ${A_n}^{*}_{ij}$ is only different by sign of the first index ${(-1)^{i+j}}$. Therefore, $B_{ij}$ which is computed by product of row $i$ in matrix ${A_n}_{ij}$ and column $j$ in matrix  ${A_n}^{*}_{ij}$, can be computed by:
 $$B_{ij}= \sum_{k=1}^{k=n+1}{{A_n}_{ik}{A_n^*}_{kj}}$$
 
Note that the ${A_n}_{ik}$ is the probability from state $i$ ``good" channels (with $i-1$ bit ``1" and $n-i+1$ bit ``0") to medium state has $k$ ``good" channels (with $k-1$ bit ``1" and $n-k+1$ bit ``0"). Moreover, if the sign is ignored, then  ${A_n}^{*}_{kj}$ also is the probability going from state $k$ to state $j$, too. However, the state $k$ includes $C{n\choose k-1}$ sub-states which have a same number of ``good" and ``bad" channels.  For example with $n=2$, state $k=2$ includes two sub-states that contains one ``good" and one ``bad" channels are ``10" an ``01". Therefore, the total number of sub-states while $k$ runs from $1$ to $n$ is $\sum_{k=1}^{k=n+1}C{n\choose k-1}= 2^n$ sub-states.  Let compute $B_{ij}$ by divided into two subsets:

\textbf{ Compute $B_{ij}$ for $i=j$}: This means that $B_{ii}$ is the sum of the probability from state $i-1$ bit ``1" go to  states has $k-1$  bit ``1"  then come back to state has $i-1$ bit ``1". In $2^n$ sub-states, we can divide back to $n+1$ categories by the number of different position between $i$ and $k$.

$\bullet$ If all the bit in $i$ and $k$ are the same, then the probability is:
{\small  $$C{n\choose 0}(1-\alpha)^{n}(1-\alpha)^n=C{n\choose 0}(1-\alpha)^{2n}$$}

$\bullet$ If all the bit in $i$ and $k$ different at only one position, then the probability is:
{\small $$C{n\choose 1}(1-\alpha)^{2(n-1)}(1-\alpha)^2$$}

$\bullet$ If all the bit in $i$ and $k$ different at only two positions, then the probability is:
{\small $$C{n\choose 2}(1-\alpha)^{2(n-2)}(1-\alpha)^{2\times2}$$}

$\bullet$ If all the bit in $i$ and $k$ different at all positions, then the probability is:
{\small $$C{n\choose n}(1-\alpha)^{2n}$$}

Therefore, $B_{ii}$ can be determined by the probability of all $n+1$ categories such as:
{\small $$B_{ii}=\sum_{t=0}^{t=n}C{n\choose t} \alpha^{2t} (1-\alpha)^{2n-2t}=({(1-\alpha)^{2} -\alpha^{2}})^{n}=(1-2\alpha)^{n}$$}
\textbf{Compute $B_{ij}$ for i$\neq$ j}: Let  divide ${A_n^*}_{kj}$ into two subsets: $k+j$ is odd  and ${A_n^*}_{kj} < 0$ or $k+j$ is even  and ${A_n^*}_{kj} > 0$, respectively. Therefore, $B_{ij}=\sum_{k=1}^{k=n}{{A_n}_{ik}{A_n^*}_{kj}}$ also is distributed into positive or negative subsets. Next, we will show that the positive subset in  $B_{ij}$ is equal the negative subset then $B_{ij}=0$ for $i \neq j$. Indeed, suppose that state $i$ with $i-1$ bit ``1" go to state $k_1$ and then to back to state $j$ with $j-1$ bit ``1" and $B_{ik_1}$ is positive value.  Next, we will show that existence a state $k_2$ such that $B_{ik_2}$ is negative value and $B_{ik_1}=-B_{ik_2}$. 

Let call $s$ is the number of positions where state $i$ and $j$ have a same bit. Obviously that $s\leq n-1$ due to $i\neq j$. For example if $n=4$ and $i=1111$ and $j=0001$, we have $s=1$ because $i$ and $j$ share a same bit ``1" in the positions fourth. Suppose  that an arbitrary state $k_1$ are picked, we will show how to chose the state $k_2$  with $B_{ik_1}=-B_{ik_2}$. Consider two follows cases:

$\bullet$ If $(n-s)$ is odd. $k_2$ is constructed by maintain $s$ position of $k_1$ where $i$ and $j$ have same bit and flip bit in  the $n-s$ rest positions.

$\bullet$ If $(n-s)$ is even. $k_2$ is constructed by maintain $s+1$ position of $k_1$ where $s$ position are $i$ and $j$ have a same bit and one position where $i$ and $j$ have a different bit, next $n-s-1$ rest positions will be flipped. Note that since $s\leq n-1$ then we are able to flip $n-s-1$ rest positions.

We obviously can see that $k_1$ and $k_2$ satisfied the probability condition $|B_{ik_1}|=|B_{ik_2}|$ due to the number of flipped bit between $i$ and $k_1$ equals the number of flipped bit between $k_2$ and $j$ and  the number of flipped bit between $j$ and $k_1$ equals the number of flipped bit between $k_2$ and $i$.

Next, we will prove that $k_1$ and $k_2$ make $B_{ik_1}$ and $B_{ik_2}$ in different subsets. Indeed, call number of bit ``1" in $k_1$ is $b_1$, number of bit ``1" in $k_2$ is $b_2$, number of bit ``1" in $s$ bit same of $i$ and $j$ is $b_s$, respectively. Therefore, the number of bit ``1" of $k_1$ in $(n-s)$ rest positions is $(k_1-k_s)$, the number of bit ``1" of $k_2$ in $(n-s)$ rest positions is $(k_2-k_s)$.  

$\bullet$ If $(n-s)$ is odd. Since all bit in $(n-s)$ rest positions of $k_1$ is flipped to create $k_2$, then total number of bit ``1" in  $n-s$ bit of $k_1$ and $k_2$ is $(k_1-k_s+k_2-k_s=n-s)$ is odd. So, $(k_1+k_2)$ should be an odd number.  That said $(k_1-k_2)$ is odd or $(k_1+j)-(k_2+j)$ is odd. Therefore, $B_{ik_1}$ and $B_{ik_2}$ bring the contradict sign.

$\bullet$ If $(n-s)$ is even. Because, we fix one more position to create $k_2$, then number of flipped  bit $(n-s-1)$ is odd number. If one more bit is fixed in $k_1$ is ``0", we have a same result with case $(n-s)$ is odd. If fixed bit is ``1", similarly in first case $(k_1-k_s-1)+(k_2-k_s-1)=n-s-1$ is odd number, therefore $(k1+k2)$ is odd number. That said $(k_1-k_2)$ is odd or $(k_1+j)-(k_2+j)$ is odd. Therefore, $B_{ik_1}$ and $B_{ik_2}$ bring the contradict sign.

Therefore, the state $k_2$ always can be created from a random state $k_1$ and $B_{ik_1}$ and $B_{ik_2}$ bring a contradict sign. That said for $i \neq j$, $B_{ij}=0$. Therefore:

$$B={(1-2\alpha)^{n}} I$$

The Proposition \ref{inverse}, therefore, are proven.
\end{proof}

\end{document}